\documentclass[twocolumn,9pt]{extarticle}

\usepackage{fullpage}
\usepackage{times}
\usepackage{graphicx}
\usepackage{balance}
\usepackage{listings}
\sbox0{\small\sffamily x}
\edef\mybasewidth{\the\wd0 }
\lstdefinestyle{custompython}{
  language=Python,
  basicstyle=\small\sffamily,
  keywordstyle=\bfseries\color{blue!40!black},
  commentstyle=\itshape\color{green!40!black},
  columns=fixed,basewidth=\mybasewidth,
  morekeywords={None,True},
  belowskip=0pt,
  belowcaptionskip=-10pt,
  breaklines=true}

\usepackage{xcolor}
\usepackage{amsmath}
\usepackage{paralist}

\usepackage{booktabs}
\usepackage{amsthm}
\newtheorem{theorem}{Theorem}
\newtheorem{definition}{Definition}

\def\BibTeX{{\rm B\kern-.05em{\sc i\kern-.025em b}\kern-.08emT\kern-.1667em\lower.7ex\hbox{E}\kern-.125emX}}
    
\renewcommand{\paragraph}[1]{\smallskip\noindent{\em\bfseries #1}}

\newcommand{\fulltitle}{Spatial Locality and Granularity Change in Caching}

\newcommand{\fullproblem}{Granularity-Change Caching Problem}
\newcommand{\fullprob}{Granularity-Change Caching}
\newcommand{\problem}{GC Caching Problem}
\newcommand{\prob}{GC Caching}

\newcommand{\pr}{GC}
\newcommand{\offproblem}{Offline Granularity-Change Caching Problem}
\newcommand{\offprob}{Offline GC Caching}

\newcommand{\fullmodel}{Granularity-Change Caching Model}
\newcommand{\element}{item}
\newcommand{\block}{block}
\newcommand{\bsizemath}{B}
\newcommand{\bsize}{$\bsizemath{}$}
\newcommand{\fulltimemath}{a}
\newcommand{\fulltime}{$\fulltimemath{}$}

\newcommand{\opt}{optimal cache}
\newcommand{\elempol}{Item Cache}
\newcommand{\blockpol}{Block Cache}
\newcommand{\ourpolicyfull}{Item-Block Layered Partitioning}
\newcommand{\ourpolicy}{IBLP} 
\newcommand{\elempart}{\element{} layer}
\newcommand{\blockpart}{\block{} layer}
\newcommand{\epsizemath}{i}
\newcommand{\bpsizemath}{b}
\newcommand{\epsize}{$\epsizemath{}$}
\newcommand{\bpsize}{$\bpsizemath{}$}
\newcommand{\ourrandpolicyfull}{Granularity-Change Marking}
\newcommand{\ourrandpolicy}{GCM}

\begin{document}

\title{\textbf{\fulltitle}}


\author{{Nathan Beckmann}\\
{beckmann@cs.cmu.edu}\\
{Carnegie Mellon University} \and
{Phillip B. Gibbons}\\
{gibbons@cs.cmu.edu}\\
{Carnegie Mellon University} \and
{Charles McGuffey}\\
{cmcguffey@reed.edu}\\
{Reed College}}

\date{}


%
\maketitle

\begin{abstract}

Caches exploit temporal and spatial 
locality to allow a small memory to provide fast access to data
stored in large, slow memory.
The temporal aspect of locality is extremely well studied and understood,
but the spatial aspect much less so.
We seek to gain an increased understanding of spatial locality by 
defining and studying the \emph{\fullproblem{}}. 
This problem modifies the traditional caching setup by grouping data 
\element{}s into \block{}s, such that a cache can choose \emph{any subset}
of a \block{} to load 
for the same cost as loading any individual \element{} in the \block{}.

We show that modeling such spatial locality significantly changes the caching problem. 
This begins with a proof that \fullprob{} is NP-Complete in the offline setting, 
even when all \element{}s have unit size and all \block{}s have unit load cost. 
In the online setting, we show a lower bound for competitive ratios of 
deterministic policies that is significantly worse than traditional caching. 
Moreover, we present a deterministic replacement policy called 
\emph{\ourpolicyfull{}} and show that it obtains a competitive ratio close to 
that lower bound.
Moreover, our bounds reveal a new issue arising in the \fullproblem{},
where the choice of offline cache size affects the 
competitiveness of different online algorithms relative to one another.
To deal with this issue, we extend a prior (temporal) locality model to account for spatial 
locality, and provide a general lower bound in addition to an upper bound for
\ourpolicyfull{}.

\end{abstract}

\section{Introduction}
\label{sec:intro}

A common feature of computer systems is that \emph{block granularity} changes at different levels of the memory/storage hierarchy.
This paper presents the first study of how granularity change affects caching. 
We define the \emph{Granularity-Change (GC) Caching Problem},
prove new adversarial competitive bounds for the problem,
develop an online GC caching policy with a better 
competitive ratio than traditional cache policies in this setting,
define a new locality model for the problem,
and prove upper and lower bounds within this locality model.

\paragraph{Why does block granularity change?}
Given that a large and fast memory does not exist, real systems make use of a 
hierarchy ranging from small, fast memories to larger and slower storage 
devices~\cite{hennessy2011computer},
typically structured as a hierarchy of caches.
Cache hierarchies provide the illusion of a large, fast memory by
exploiting temporal and spatial locality in data accesses.
\emph{Temporal locality} refers to when the same data item is
referenced multiple times in quick succession.
\emph{Spatial locality} refers to when nearby data items are
referenced in quick succession.

Each level of the storage hierarchy organizes its data in blocks 
to simplify management and reduce overheads.
For example,
SRAM caches typically consist of 64\,B ``lines'',
DRAM of 2-4\,KB ``rows'',
and flash/disk of 4\,KB ``pages''.\footnote{In fact, there can be 
different granularities for reads and writes, 
e.g., ``erase blocks'' in flash can be many MBs. We focus on reads in this work.}
Organizing data in blocks is a simple way for caches to exploit spatial locality,
as each block can typically hold multiple data items.

\begin{figure}
\begin{center}
\includegraphics[width=0.6\linewidth]{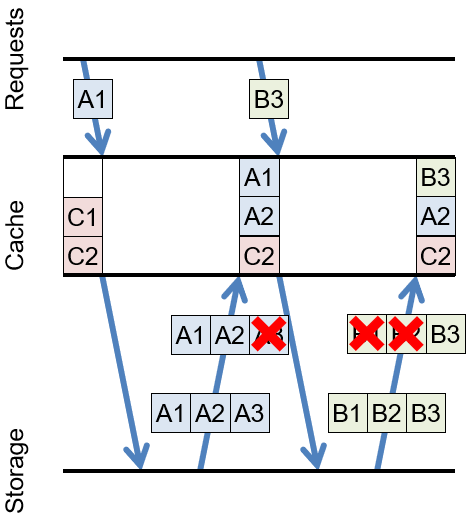}
\end{center}
\vspace{-0.05in}
\caption[Opportunity of the \problem{}]{ In the \problem{}, caches can
  load any subset of the larger-granularity block from the level
  below them, for the same cost.  Here, although only A1 is requested, 
  the cache loads the subset \{A1,A2\} of the block \{A1,A2,A3\}.}
\label{fig:eye-candy}
\vspace{-0.1in}
\end{figure}

Most caches today ignore granularity change and only load data of their own 
granularity.
But this misses an optimization opportunity to load some or all of the larger-granularity block
at minimal cost (see Figure~\ref{fig:eye-candy}),
as the lower level has already fetched the entire block.
Motivated by this observation, we ask:
\emph{What caching opportunities and challenges are introduced by granularity change?}

\paragraph{What does prior caching work say about block granularity?}
The original caching problem is well studied and understood. 
Belady~\cite{belady1966study} and Mattson~\cite{mattson1970evaluation}
separately devised optimal solutions for caching with unit-size, unit-cost 
\element{}s. 
Sleator and Tarjan~\cite{sleator1985amortized} provided both lower and upper 
bounds for the cost ratio when 
comparing the performance of \emph{online} caches, which must make decisions 
as requests arrive, against \emph{offline} caches, which are allowed to view 
the entire trace when making decisions. 
Fiat et al.~\cite{fiat1991competitive} extended this work to randomized 
algorithms and showed ways of approximating
online policies using other online policies.

Other variants of caching have been considered, and considerable work 
has been done on complexity and algorithms for these variants~\cite{
Chrobak90newresults,
albers1999page,
bar2001unified,
brehob2004optimal,
young1994thek, even2018online}.
Caching with variable-size \element{}s~\cite{chrobak2012caching, berger2018foo, young2002line} appears to be 
similar to \prob{}, since one could think of different subsets of
a \block{} as differently sized \element{}s.
The critical difference is that, unlike in variable-size caching,
\element{}s in a \block{} can be accessed, cached, and evicted individually.

In \prob{}, choosing which \element{}s to load is an 
additional dimension with significant impact on performance.
To our knowledge, there is no prior theoretical work that accounts for the granularity change optimization discussed above.
Prior caching work focuses on the temporal locality of \element{}s within the access trace (where an \element{} can be of fixed or variable size), and misses the significant impact of spatial locality among data items in the level below.
Models of computation that account for transfers to and from the cache in ``blocks''
(e.g., the \textit{External Memory} model~\cite{AggarwalV88},
the \textit{Ideal Cache} model~\cite{frigo99}, the \textit{Multicore-Cache}
model~\cite{BCGRCK08}, and the \textit{Parallel Memory Hierarchy}
model~\cite{alpern1993pmh}) permit \element{}s in a block to be individually
accessed from cache, but not individually cached or evicted.
As such, the ``blocks'' are the (smaller) granularity of the cache itself (e.g., A1 in Figure~\ref{fig:eye-candy}) and not the (larger)
granularity of the level below (A1 A2 A3), as in our model.


The systems community uses several approaches to handle
granularity change. 
There is work on scheduling memory controllers at granularity boundaries~\cite{
fan2001memory,
mutlu2007stall,
mutlu2008parallelism,
zheng2008mini,
zhu2005performance, zhu2002fine}, address mapping techniques~\cite{
sudan2010micro,
lin2001designing,
zhang2000permutation, yoon2012row}, row-buffer management~\cite{
miura2001dynamic,
park2003history,
stankovic2005dram,xu2009prediction}, and \element{}-to-\block{}
allocation~\cite{calder1998cache, 
chilimbi1999cache, attiya2018remote, petrank2002hardness}.
Recently, DRAM caches account for granularity 
change by taking some or all of the larger-granularity block into the smaller-granularity cache on 
loads~\cite{qureshi2012fundamental,jevdjic2013stacked,jevdjic2014unison}.
We provide the first theoretical framework to better understand 
and guide these designs.

\paragraph{Contributions.}
We investigate the effects of granularity change on caching. 
Our results include:
\begin{compactenum}
\item[\emph{(i)}] We develop a model for caching at a granularity boundary, 
  called the \fullproblem{};
\item[\emph{(ii)}] We show that \offprob{} is NP-Complete; 
\item[\emph{(iii)}] We provide a lower bound on the competitive ratio of deterministic
replacement policies in \pr{} Caching; 
\item[\emph{(iv)}] We design and analyze \ourpolicyfull{}, a practical GC caching policy, 
and we prove an upper bound that is much 
tighter than any policy that considers only a single granularity (Table~\ref{tab:int-points} highlights our bounds);
\item[\emph{(v)}] We discuss a new issue that arises in \prob{} where
the choice of offline cache size affects the 
competitiveness of different online algorithms relative to one another, and
we develop a locality model for \prob{} that admits upper and lower bounds
for fault rate based solely on the system's 
cache size and the workload.
\end{compactenum}

\section{The Model}
\label{sec:model}

The \fullmodel{} consists of a single level of memory (cache) that receives a 
series of requests, referred to as accesses, to data \element{}s. 
If the \element{} is in the cache, then the request is served and the cache 
is not charged. 
If the \element{} is not in the cache, then the cache must load the \element{} 
from the subsequent level of memory or storage.
If this load causes the amount of data in cache to exceed the cache size $k$, 
then \element{}s must be evicted from the cache to remedy the situation. 

What makes the \fullmodel{} unique is that the universe of \element{}s is 
partitioned into \emph{\block{}s} of up to \bsize{} \element{}s, such that
the cache can load \emph{any subset} of the \element{}s in a \block{} 
for unit cost;
i.e., \element{}s after the first are ``free'' (\bsize=3 in Figure~\ref{fig:eye-candy}).
When each \element{} is in a different \block{}, this model 
exactly matches the traditional caching model. 

The \block{}s represent the larger data granularity used by the subsequent level of 
the memory hierarchy. 
In such systems, there is typically a small memory buffer used to 
handle data as it is being read or written. 
The cost of moving data from the subsequent level into this buffer is typically 
large relative to the cost of operating on the buffer itself. 
Hence, once \element{}s are brought into the buffer,
they can be accessed at low cost, motivating our model~\cite{hennessy2011computer,
  jacob2010memory}. 

\begin{definition}
  In the \emph{\fullproblem}, we are given \emph{(i)} a cache of
  size $k$, \emph{(ii)} an (online or offline) trace $\sigma$ of requests 
  to \element{}s, and \emph{(iii)} a partitioning of the \element{}s into 
  disjoint \block{}s (sets) such that no partition contains more than \bsize{} \element{}s.
  Starting with an empty cache, the goal is to minimize the number of times
  that an \element{} is not present in the cache when requested in $\sigma$. 
  When a requested \element{} is not in cache, any subset of that \element{}'s 
  \block{} can be loaded, so long as the subset contains the \element{}. 
\end{definition}

\begin{table}[t]
  \footnotesize
  \resizebox{\linewidth}{!}{%
    \newcommand{\headerfont}[1]{#1}
\newcommand{\entryfont}[1]{{\normalsize #1}}
\begin{tabular} {@{}c|ccc@{}}
\toprule
\bf Setting             & \bf Sleator-Tarjan Bound       & \bf Our GC Lower Bound                                              & \bf Our GC Upper Bound                                               \\
\midrule
\headerfont{Constant Augmentation} & \entryfont{$k=2h \Rightarrow 2\,\times$}   & \entryfont{$k \approx 2h \Rightarrow \bsizemath{}\,\times$}                  & \entryfont{$k \approx 2h \Rightarrow 2\bsizemath{}\,\times$}                  \\
\headerfont{Ratio = Augmentation}  & \entryfont{$k = 2h \Rightarrow 2\,\times$} & \entryfont{$k \approx{} \sqrt{\bsizemath{}}h \Rightarrow \sqrt{B}\,\times$} & \entryfont{$k \approx{} \sqrt{2\bsizemath{}}h \Rightarrow \sqrt{2B}\,\times$} \\
\headerfont{Constant Ratio}        & \entryfont{$k=2h \Rightarrow 2\,\times$}   & \entryfont{$k \approx{} \bsizemath{}h \Rightarrow 2\,\times$}                & \entryfont{$k \approx{} \bsizemath{}h \Rightarrow 3\,\times$}                 \\
\bottomrule
\end{tabular}
    }
  \caption[Salient bounds]{Salient bounds
    for online cache size $k$ and optimal cache size $h$,
    shown as: Augmentation $\Rightarrow$ Competitive Ratio.
Compared to traditional caching (Sleator-Tarjan Bound), 
the spatial locality in \pr{} Caching adds a penalty of $\Theta{}(\bsizemath{})\times{}$ 
to the product of competitive ratio $\times$ augmentation.}
\label{tab:int-points}
\end{table} 

\paragraph{Locality vs.\ traditional caching models.}
In traditional caching models, all hits come from \emph{temporal locality},
i.e., when an \element{} remains in cache between subsequent accesses.
In \pr{} Caching, hits can also come from \emph{spatial locality},
i.e., when an \element{} $I$ is in cache due to an earlier access to a \emph{different} \element{} in the same \block{}.
(Any hits to item $I$ beyond the first are due to temporal locality, since $I$ would have been brought in cache anyway.)

\paragraph{Assumptions.}
We assume that each \element{} has unit size and each \block{} has unit cost. 
We also assume that the caches we study are much larger than 
the \block{} size, i.e., $k \gg \bsizemath{}$. 
In addition, we limit our results to deterministic policies. 

\paragraph{Baseline policies.}
We consider two baseline cache designs.
An \emph{\elempol{}} loads only the requested \element{} from a \block{};
i.e., it is a traditional cache.
By contrast, a \emph{\blockpol{}} loads all the \element{}s in a requested block
and also evicts them together;
i.e., it increases the cache's granularity to operate on \block{}s instead of \element{}s.
\elempol{}s perform well on temporal locality and poorly on spatial locality, 
whereas \blockpol{}s are the opposite. 

\paragraph{Locality Model.}
To expand our analysis beyond competitive ratios on cost, we extend the locality of reference model by Albers, Favrholdt, and Giel~\cite{albers2005paging} to account for block granularity. 
Their model adds a function $f(n)$ (which will increasing and concave for real traces) that characterizes the maximum size of a working set (the number of distinct \element{}s accessed) in a window of $n$ accesses across a trace.
They use this model to provide bounds on the fault rate (the number of faults per access) of various replacement policies as a function of $f(n)$. 

We extend this model to account for block granularity by adding a similar function $g(n)$ to account for the number of distinct \emph{\block{}s} accessed in a window of size $n$. 
The value of $g(n)$ can range from $f(n)$ when each \element{} comes from a different \block{} to $f(n) / \bsizemath{}$
when entire \block{}s are accessed at a time.
The ratio $f(n) / g(n)$ measures how much spatial locality occurs in a given trace.
With this extended model, we can provide bounds on the fault rate of policies in the \fullproblem{}.

\section{Complexity Analysis}
\label{sec:complexity}

In this section, we investigate the complexity of \prob{}.
Using a reduction from variable-size caching, we are able to show that 
the \offprob{} is NP-Complete, even with unit \block{} cost and 
unit \element{} size. 

\begin{theorem}
The \offproblem{} is NP-Complete.
\vspace{-0.1in}
\end{theorem}

\begin{proof}
Our proof relies on a reduction from variable size caching, which is known 
to be NP-complete~\cite{chrobak2012caching}. 
We begin by showing how to create \element{}s for the \problem{} and 
then assign them to blocks. 
We then use these blocks to generate a trace where the cost paid by the 
\opt{} is the same as the optimal cost for the variable-size caching problem. 

The first step of the reduction is to scale the variable-size caching problem 
to have integral \element{} sizes. 
This can be done by multiplying the size of each \element{} and the cache 
size by the same value (assuming the sizes are rational numbers). 
After the sizes are all integral, the \element{}s of the \problem{} can be 
created and partitioned. 
For each \element{} in the variable-size caching problem, we create one 
\block{} in \prob{}. 
The maximum size of these \block{}s can be chosen to be any value greater than 
or equal to the size of the largest \element{} in the scaled variable-size 
problem. 
For each of these \block{}s we will use only the first $z$ \element{}s, where 
$z$ is the size of the corresponding \element{} in the scaled problem. 
We refer to these as the \emph{active set} for that \block{}. 

The idea for trace generation is to create a trace for \prob{} that has 
accesses to the same amount of cache space as in the variable-size problem. 
For each access in the variable size trace, the \prob{} trace replaces it with 
consecutive accesses to the active set of the corresponding \block{}. 
Each \element{} is accessed a number of times equal to the number of 
\element{}s in the active set, in round robin order. 
The ordering of the variable-size trace is maintained, so that the ordering of 
the \block{}s in the \prob{} trace is the same as the ordering of the 
\element{}s in the variable-size trace. 
The cache size is set to be the same as that of the scaled variable-size 
problem. 

We are left to show that the optimal cost of the \problem{} that we generate 
is the same as the optimal cost of the variable-size caching problem. 
First, we show that scaling sizes does not affect the result. 
Since the cache size was scaled by the same factor as the \element{}s, the 
fraction of the cache space that each \element{} takes remains unchanged. 

Beyond this, we show that there is an optimal solution for the generated 
\prob{} instance that loads and evicts the entire active set of a \block{} 
at the same time. 
To do this, we rely on the fact that an optimal solution can load all \element{}s 
from the active set that are not in cache for unit cost. 
This means that any consecutive series of requests to a single \block{} can be 
served for that unit cost. 
However, unless the cache contains the entire active set, it must 
pay at least unit cost to serve the \block{}. 

Since the active sets corresponding to different \element{}s in the 
variable size problem{} are in distinct \block{}s, there is no way to load an 
active \element{} without paying unit cost for that \block{}. 
When combined with the fact that the active \element{}s remain unchanged for 
each \block{}, we can show that evicting a single active \element{} 
increases the cost paid by the cache the same as evicting all active 
\element{}s for that \block{}. 

If we assume the cache evicts all active \element{}s for a \block{} at the 
same time, then upon the first consecutive access to a \block{}, either the 
entire active set will be in cache or none of it will be. 
If the set is in cache, no loads are required. 
Otherwise, loading less than the entire set will cause
the need for additional loads. 
In particular, the maximum amount of cache space used for the active set 
multiplied by the number of loads must exceed the active set. 

We use the repeated accesses to the active set to show that the benefits of 
such additional loads cannot outweigh the drawbacks. 
Since any \element{} can be loaded for unit cost, the benefit of using less 
cache space for one set of consecutive accesses to the active set cannot exceed 
the amount of cache space reserved. 
Due to the repeated nature of the accesses, loading less than the entire active 
set will result in at least as many additional loads as the size of the active set. 
This means that loading the entire active set upon the first miss is optimal. 

Since we have shown that an optimal solution is to load and evict entire active 
sets at a time, any access that immediately follows an access to the same 
active set will automatically be a hit. 
Putting these results together, the optimal solution to the generated trace is 
the same as the optimal solution to the original trace. 
\end{proof}

\begin{figure}
\includegraphics[width=\linewidth]{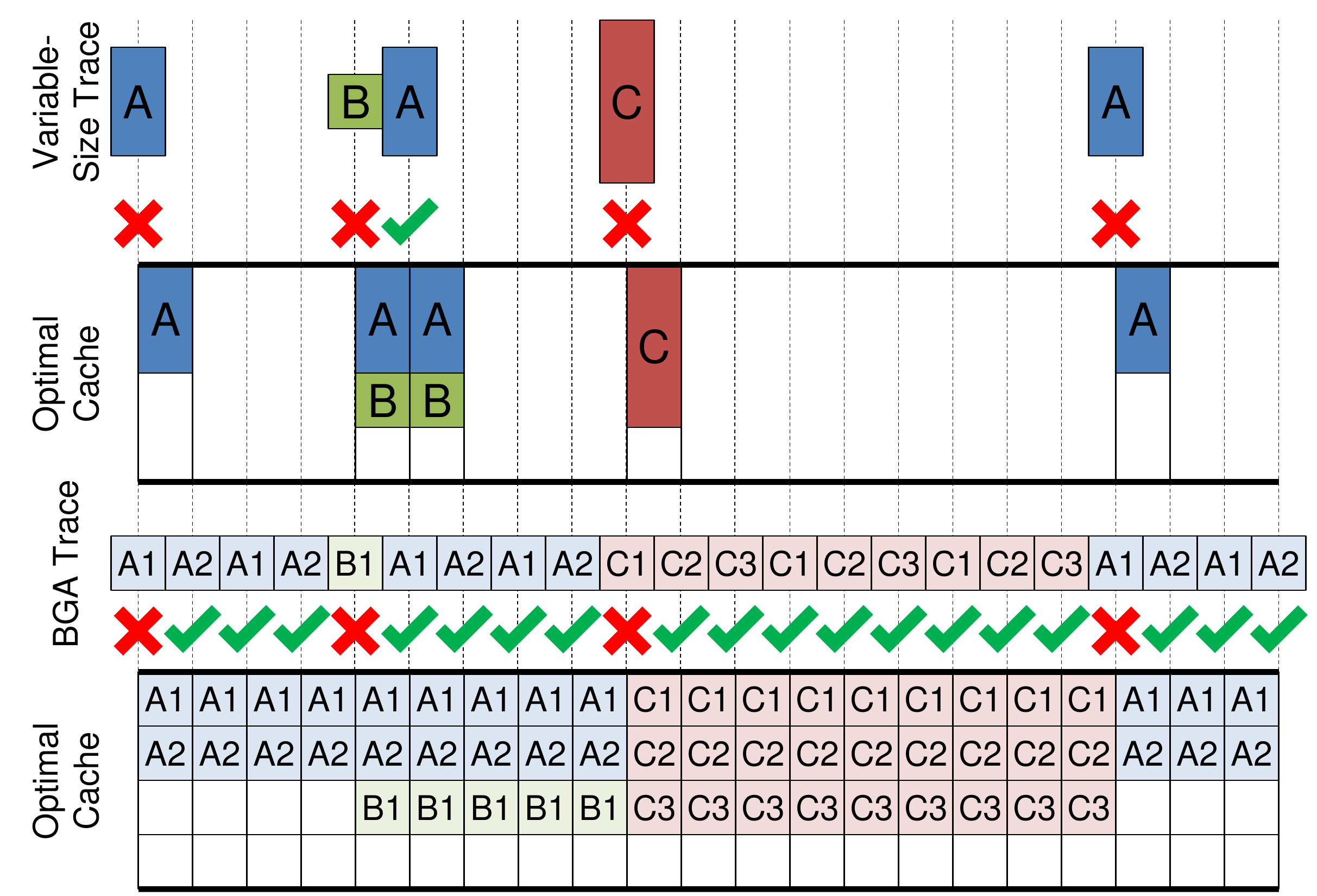}
\vspace{-0.15in}
\caption{
A diagram showing how to transform an instance of variable-size caching 
into an instance of the \prob{} with the same cost. 
Notice that storing \element{} A1 with \block{} C would not decrease the 
number of misses, since A2 would still suffer a miss. 
The repeated accesses force the optimal solution to load the entire \block{} 
in order to minimize misses. 
}
\label{fig:reduction-diagram}
\end{figure}

Figure~\ref{fig:reduction-diagram} shows an example of the reduction. 
We simulate variable size \element{}s through the use of multiple 
\element{}s from the same \block{} accessed consecutively. 
By repeating these access sequences, we force the optimal solution to load 
all \element{}s that are used in the \block{}. 
The optimal solution to the input instance can easily be translated into an 
optimal solution to the generated instance.

\section{Competitive Lower Bound}
\label{sec:bacs-comp-lb}

We next show how to adapt prior techniques for
competitive ratios to the \problem{}. 
We will start by providing a lower bound for \elempol{}s and \blockpol{}s, 
and then use the insights we gain to devise a more general lower bound. 

\subsection{\elempol{}s} \label{sec:item-bound}
We start by examining how \elempol{}s perform in the \problem{}. 
These policies are widespread and well studied, so they serve as a logical 
starting point for our investigation. 

The traditional lower bound for \elempol{}s comes from Sleator and 
Tarjan~\cite{sleator1985amortized}, and is defined as follows.
We define $k$ to be the size of the online cache and $h$ to be the size of 
the optimal cache. 
\begin{enumerate}
\item Assume both the optimal and online caches are full. 
\item Access $k - h + 1$ \element{}s that have not been seen before.
\item Create a set of \element{}s containing the \element{}s in the optimal 
cache during step one and the \element{}s accessed during step two. 
This set contains $k + 1$ \element{}s. 
\item Access the \element{} from the set that is not in the online cache. 
Repeat this process a total of $h-1$ times. 
\item To generate a longer trace, return to step one. 
Note that the assumption will be satisfied. 
\end{enumerate}
In these traces the online policy never hits, as the \element{}s in step two 
are newly accessed and the \element{}s in step four are chosen to be outside 
of that cache. 
The optimal policy also misses on every access in step two. 
Since it knows the future accesses, it can perform evictions so as to 
store each \element{} that will be accessed in step four, allowing it to achieve 
hits on all of these accesses. 

The \problem{} problem modifies traditional caching by 
introducing spatial locality. 
This means that the optimal policy can miss on only one out 
of every \bsize{} accesses in step two if the \element{}s are chosen such that 
an entire \block{} is accessed. 
The cost of this modification is that the optimal cache uses \bsize{} space 
rather than unit space in step two, 
and thus step four must be shortened to  $h-\bsizemath{}$ accessses. 
This provides the following result:

\begin{theorem}\label{thm:lb-elempol}
Any \elempol{} has a competitive ratio of at least 
$\bsizemath{} (k-\bsizemath{}+1) / (k-h+1)$ 
where $k$ is the size of the online cache and 
$h$ is the size of the optimal cache.
\end{theorem}

\begin{proof}
  Create a trace according to the following procedure.
  Steps (1), (3), and (5) are the same as above for traditional caching. The other two steps are modified slightly:
\begin{itemize}
\item[(2)] Choose a \block{} that has not been seen before. 
Access each \element{} from that \block{}. 
Repeat this process until $k-h+1$ \element{}s have been accessed 
in this step. 
\item[(4)] Access the \element{} from the set that is not in the online cache. 
Repeat this process a total of $h-\bsizemath{}$ times. 
\end{itemize}

Since the online policy does not load any \element{} that is not accessed, it 
will miss on each access in step two. 
The accesses in step four are chosen to ensure that the online policy misses on 
each. 

The optimal policy will load each \block{} on its first access in step two, 
resulting in $(k-h+1) / \bsizemath{}$ misses. 
It can use its remaining $h-\bsizemath{}$ space to store the \element{}s 
accessed in step four, hitting on each. 

Taking the ratio of these misses provides the bound. 
\end{proof}

This competitive ratio is nearly a multiplicative \bsize{} factor worse than 
traditional lower bounds under the assumption that $k \gg \bsizemath{}$. 
Since the behavior of the policies
has not changed, this shows the increased power of the optimal algorithm in this 
model. 

\subsection{\blockpol{}s}\label{sec:block-pol-lb}

If \elempol{}s do not perform well in the \problem{}, perhaps \blockpol{}s will. 
They are naturally suited to handle the access sequences  
\elempol{}s suffered on, and may therefore offer a better realization of the 
potential performance improvements. 

This intuition proves accurate for the trace design scheme described in
Section~\ref{sec:item-bound},
Online \blockpol{}s will be able to hit on all subsequent accesses to a \block{} 
in step two, providing the same number of misses as the optimal policy on that 
step. 
As long as the number of accesses in the second step dominates the number in 
the third, these policies will perform well. 

The problem with \blockpol{}s is that when only a few \element{}s of a block 
are accessed, the remaining \element{}s will pollute the cache, reducing the 
available space to serve accessed \element{}s. 
In particular, when only one \element{} from a \block{} is used at a time, the 
cache is effectively \bsize{} times smaller. 
We can apply this insight along with the traditional bound to provide a bound 
for \blockpol{}s. 

\begin{theorem}\label{thm:lb-blockpol}
Any \blockpol{} has a competitive ratio of at least 
$k / (k-\bsizemath{}(h-1))$
where $k$ is the size of the online cache and 
$h$ is the size of the optimal cache.
\end{theorem}

\begin{proof}
Create a trace according to the following procedure:
\begin{enumerate}
\item Assume both the optimal and online caches are full and 
that each \element{} in the optimal cache is from a different \block{}. 
\item Access one \element{} from each of 
$\lceil{} k / \bsizemath{} \rceil{} - h + 1$ 
\block{}s that have not yet been accessed. 
\item Create a set of \element{}s containing the \element{}s in the optimal 
cache during step one and the \element{}s accessed during step two. 
This set contains $\lceil{} k / \bsizemath{} \rceil{} + 1$ \element{}s, 
each of which are from a different \block{}. 
\item Access the \element{} from the set that is not in the online cache. 
Repeat this process a total of $h-1$ times. 
\item To generate a longer trace, return to step one. 
Note that the assumption will be satisfied. 
\end{enumerate}

The online policy has not seen any of the \block{}s accessed 
in step two, so it will miss on each access. 
Because it loads and evicts at \block{} granularity, it will store exactly 
$\lceil{} k/\bsizemath{} \rceil{}$ \block{}s at a time. 
This means that step four can always choose one \element{} that is not in 
the online cache, resulting in misses for all accesses in that step. 

The optimal policy will miss on each access in step two 
for $\lceil{} k / \bsizemath{} \rceil{} - h + 1$ misses. 
Since it knows what will be accessed in step four, it can 
choose to keep those \element{}s and hit on each such access. 

With the assumption that \bsize{} evenly divides $k$, 
taking the ratio of the misses and simplifying results in the target bound. 
\end{proof}

This result shows that \blockpol{}s perform very poorly on traces that do not 
take advantage of the spatial locality that is available in the \problem{}. 
In particular, they suffer a performance penalty where the effective cache size 
is reduced by a factor equal to the ratio between the 
size of the \block{} and the average number of \element{}s used per \block{}.
This means that the competitive ratio of such policies is infinite unless they 
have at least \bsize{} times as much space as the optimal policy to which they are being 
compared. 

\subsection{Generalizing the Lower Bound}

Building on the above discussion, we now provide a more general lower bound for policies beyond
\elempol{}s and \blockpol{}s. 

Like the worst-case traces for \elempol{}s and \blockpol{}s, we will first 
access new \element{}s until the online cache can no longer store the entire 
active set, and then repeatedly access whichever \element{} the cache chooses 
not to store. 
Unlike the previous bounds, we cannot make assumptions about the granularity 
of loads or evictions. 

When adding new \element{}s to the active set, we take advantage of 
spatial locality
to allow the \opt{} to outperform the online cache. 
For any \block{}, the worst-case trace can continue to access an \element{} from 
that \block{} that is not in the online cache until no such \element{} exists. 
The online cache will miss on each access, while the \opt{} can load each 
accessed \element{} on the first access to the \block{}. 
Using this insight, we categorize policies using a parameter \fulltime{}, 
which is the number of distinct consecutive accesses to a \block{} 
before the policy loads the entire \block{}. 

Our trace construction 
follows what we use for \elempol{}s, 
except that the online policy will incur \fulltime{} misses 
on each block in step two and the optimal policy will need to use \fulltime{} 
space to service these requests, leaving only $h-\fulltimemath{}$ space 
available for step three. 

\begin{theorem}\label{thm:lb-general}
For any deterministic replacement policy that requires \fulltime{} consecutive 
distinct accesses to a block to load all of it, the competitive ratio of that 
policy is at least 
$(\fulltimemath{}(k-h+1) + \bsizemath{}(h-\fulltimemath{})) / (k-h+1)$
where $k$ is the size of the online cache and 
$h$ is the size of the optimal cache.
\end{theorem}

\begin{proof}
Create a trace according to the following procedure:
\begin{enumerate}
\item Assume both the optimal and online caches are full. 
\item For $\lceil{} (k-h+1) / \bsizemath{} \rceil{}$ 
\block{}s that have not yet been accessed: \\
While there exists an \element{} from that block that the online cache 
has not yet loaded, access that \element{}. 
This will occur \fulltime{} times per \block{}. 
\item Create a set of \element{}s containing the \element{}s in the optimal 
cache during step one and the \element{}s in the \block{}s accessed during 
step two. 
This set contains at least $k + 1$ \element{}s. 
\item Access an \element{} from the set that is not in the online cache. 
Repeat this process a total of $h-\fulltimemath{}$ times. 
\item To generate a longer trace, return to step one. 
Note that the assumption will be satisfied. 
\end{enumerate}

As with our other traces, the online policy will miss on every access, 
causeing $\fulltimemath{} \lceil{} (k-h+1) / \bsizemath{} \rceil{}$ 
misses in step two and $h-\fulltimemath{}$ in step four. 
The optimal policy will load each of the \fulltime{} \element{}s that will be 
accessed in a \block{} on its first access in step two, 
resulting in $\lceil{} (k-h+1) / \bsizemath{} \rceil{}$ misses. 
It can use its remaining $h-\fulltimemath{}$ space to store the \element{}s that will 
be accessed in step four, hitting on each. 

With the assumption that \bsize{} evenly divides $k-h+1$, 
taking the ratio of the misses and simplifying results in the target bound. 
\end{proof}

\subsection{Analysis and Discussion}

Having devised lower bounds for deterministic policies for the 
\problem{}, we can use the insights we have developed to learn more about 
the original problem. 

\paragraph{Designing Policies.}
In order to minimize the lower bound, we need to consider the \fulltime{} 
parameter. 
In Theorem~\ref{thm:lb-general}, the \fulltime{} parameter shows up in one 
positive term and one negative term, both in the numerator. 
When $k-h+1 > \bsizemath{}$, then the positive term dominates, and 
minimizing \fulltime{} also minimizes the competitive ratio. 
Otherwise, the negative term dominates and maximizing \fulltime{} minimizes 
the ratio.
Thus, to achieve the best competitive ratio,
\emph{one should 
load either an entire \block{} or a single \element{}, and nothing in between.}

This result is true even when we allow the \fulltime{} parameter to be 
non-constant. 
To show this, we first consider a single cycle of the trace. 
The online policy suffers misses equal to the sum of its chosen \fulltime{} 
parameters across each \block{} accessed. 
The \opt{} will use space equal to the maximum number of accesses for a single 
\block{} to suffer one miss per \block{}. 
Since this space is forced to store particular \element{}s, it cannot be used 
in step four of the generation, and thus the online policy will suffer that many fewer misses. 
By applying these observations, we see that the resulting ratio is minimized when \fulltime{} is either maximized or minimized. 


Since \element{}s in a \block{} can only be distinguished by
whether they have been accessed or not, it makes sense to 
either want to bring in all of them or none of them after an access. 
However, it is less clear why an approach similar to that of the 
ski-rental problem, where the policy waits for confirmation that additional 
accesses are coming before incurring the expense of loading the block, is not
 correct. 
There appears to be a two-part answer to this question. 
The first part is the fact that the number of possible ``rentals'' is bounded by \bsize{}. 
The second is that that the decision to ``purchase'' additional \element{}s is relatively low cost, in that it can easily be changed by evicting the loaded \element{}s when a different \block{} is accessed. 

The intuition for choosing between 1 and \bsize{} for the \fulltime{} parameter 
can be found in the relative costs due to the types of locality. 
For caches where the online and offline caches are roughly equal size, 
the online cache needs as much space as possible to compete on traditional 
worst-case traces, and using cache space to serve spatial locality is more 
harmful than helpful. 
By contrast, when the online cache is much larger than the offline cache, 
the marginal benefit of the extra lines is small, making them more useful when 
devoted to handling spatial locality. 
In these situations, policies that load the entire block on access perform 
better.

We can apply a similar logic when analyzing eviction. 
The lower bound for \blockpol{}s shows that evicting at the \block{} granularity 
is inefficient. 
We need some way to choose between \element{}s in a \block{} for eviction. 
By considering whether an \element{} has been accessed,
we can differentiate \element{}s into two classes, 
preferring \element{}s that have been accessed over those that have not.

Putting this all together, we see that in order to maximize performance, 
policies should load the entire \block{} on an access 
(unless comparing against an optimal with similar size, 
where they should load individual \element{}s)
to take advantage of spatial locality, but evict \element{}s
individually to give preference to ones that have been accessed over those that have not.


\begin{figure}
\includegraphics[width=1.0\linewidth]{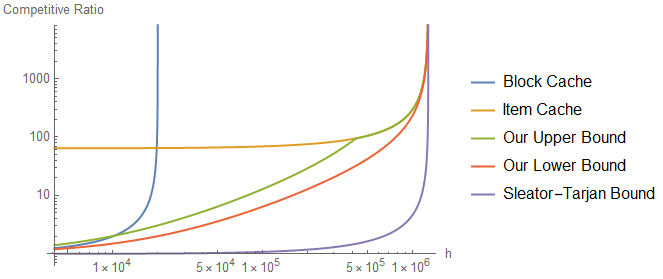}
\caption[Bounds in the \problem{}]{
  Comparing bounds in the \problem{}.
  The $x$-axis is the optimal cache size $h$, and the $y$-axis is the competitive ratio. 
  Online cache size $k = 1.28M$ and \block{} size $\bsizemath{} = 64$.
}
\label{fig:bound-comp-graph}
\end{figure}

\paragraph{The Bound in Context.}
Figure~\ref{fig:bound-comp-graph} plots the competitive ratio bounds, for a fixed online cache size $k=1.28M$ and block size $B=64$. Our resulting lower bound is much greater than 
the Sleator-Tarjan~\cite{sleator1985amortized} bound,
meaning that the gap between online and offline policies is larger in \pr{} Caching
than in traditional caching.
The gap starts at a multiplicative factor of nearly
\bsize{}$\times$ when $k \approx{} h$ 
(since the $\bsizemath{}h$ term dominates), 
and tapers off, hitting $2\times$ when $k \approx{} \bsizemath{}h$.
Table~\ref{tab:int-points} gives three salient points of comparison
for the Sleator-Tarjan bound, our lower bound, and our upper bound (discussed in Section~\ref{sec:bac-comp-ub}):
constant factor augmentation,
the point where the augmentation meets the competitive ratio,
and constant competitive ratio. 
These results show that, compared to traditional caching, 
the introduction of spatial locality increases the gap 
between online and offline policies by $\bsizemath{}\times$, which can be 
spread between the competitive ratio and the augmentation factor.
%
In prior models, the augmentation factor ($k/h$) equals the competitive ratio 
when both are $2$. 
By contrast, in the \problem{}, $k = 2h$ has a competitive ratio of 
$\approx 2 + \bsizemath{}$ and a competitive ratio of $2$ requires 
$k \approx \bsizemath{}h$. 
The meeting point of the augmentation factor and the competitive ratio occurs 
when $k \approx \sqrt{\bsizemath{}}h$. 

\section{Competitive Upper Bound}
\label{sec:bac-comp-ub}

\subsection{Policy Description}\label{sec:ub-policy-desc}

\begin{figure}
\begin{center}
\includegraphics[width=0.5\linewidth]{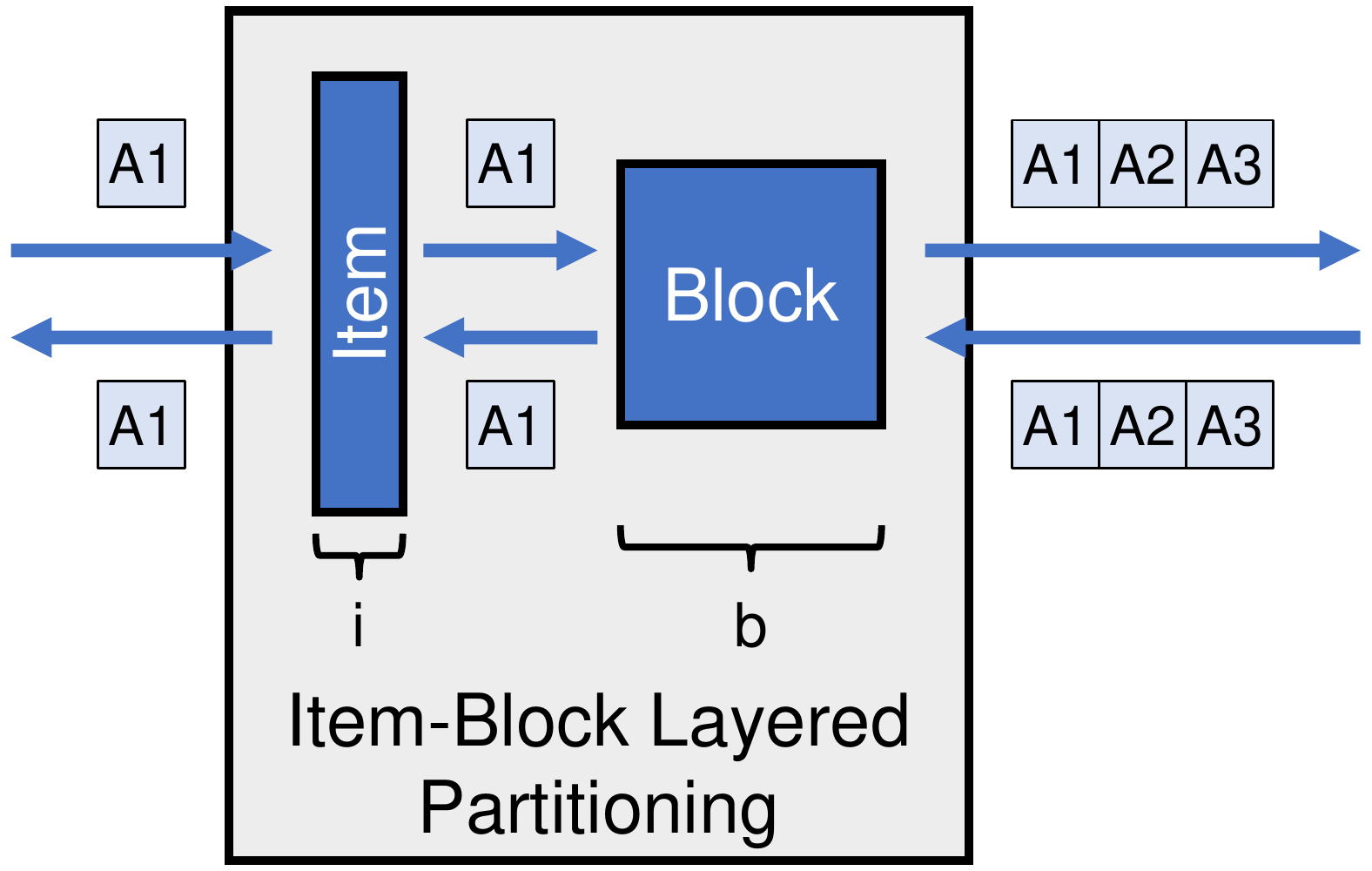}
\end{center}
\caption[A Logical Diagram of \ourpolicy{}]{
  A logical diagram of \ourpolicy{}, consisting of an \elempol{} partition 
  running LRU in front of a \blockpol{} partition running LRU. 
}
\label{fig:pol-diagram}
\end{figure}


Our policy, known as \ourpolicyfull{} (\ourpolicy{}), divides the available 
space into two layers of cache (Figure~\ref{fig:pol-diagram}).
The first layer, which serves each access to the cache, loads only the 
\element{}s that are accessed and evicts using the 
Least-Recently Used (LRU) replacement policy. 
The second layer, which only serves accesses that miss in the first layer, 
also uses the LRU policy for evictions, but loads and evicts at the 
granularity of entire \block{}s at a time.
In other words, \ourpolicy{} organizes the cache as an \elempol{}
backed by a \blockpol{}.
We refer to these layers as the \elempart{} and \blockpart{}, respectively, 
and define their sizes as \epsize{} and \bpsize{}. 

Although this is relatively straightforward to describe and 
implement, it includes some subtle design choices. 
The cache is split into two different layers to handle the two types of 
locality, 
with the \elempart{} handling temporal locality and the \blockpart{} handling 
spatial locality. 
The ordering of the two layers is important to ensure that accesses
with high temporal locality do not reorder \block{}s in the LRU list of the \blockpart{}. 
Allowing such reorderings would cause \block{}s with a small number of
frequently accessed \element{}s to pollute the \blockpart{}, reducing its 
effective space for worst-case traces. 
Note that the \blockpart{} is neither inclusive nor exclusive of the 
\elempart{}. 
If the \blockpart{} were inclusive of the \elempart{},
the \elempart{} would not contribute to the overall hit rate. 
By contrast, an exclusive policy
would avoid duplicating \element{}s, 
but would require a more complicated method of tracking \element{}s
to ensure none are evicted before their lifetimes expire in both partitions. 
Even with our simpler policy, choosing partition sizes is involved. 
We build up to it by analyzing each layer individually and then 
combining the analyses, 
with the results discussed in Section~\ref{sec:ub-app}.

\subsection{The Upper Bound}\label{sec:ub}

In order to prove our upper bound on the competitive ratio of \ourpolicy{}, 
we introduce a new linear programming technique to analyze how the \opt{}
makes use of cache space. 
Using this technique, we will first consider how each layer performs separately 
against adversarial approaches to the type of locality that it targets. 
We will then provide an analysis of the combined problem to ultimately prove 
the upper bound. 

\begin{figure}
\centerline{\includegraphics[width=0.6\linewidth]{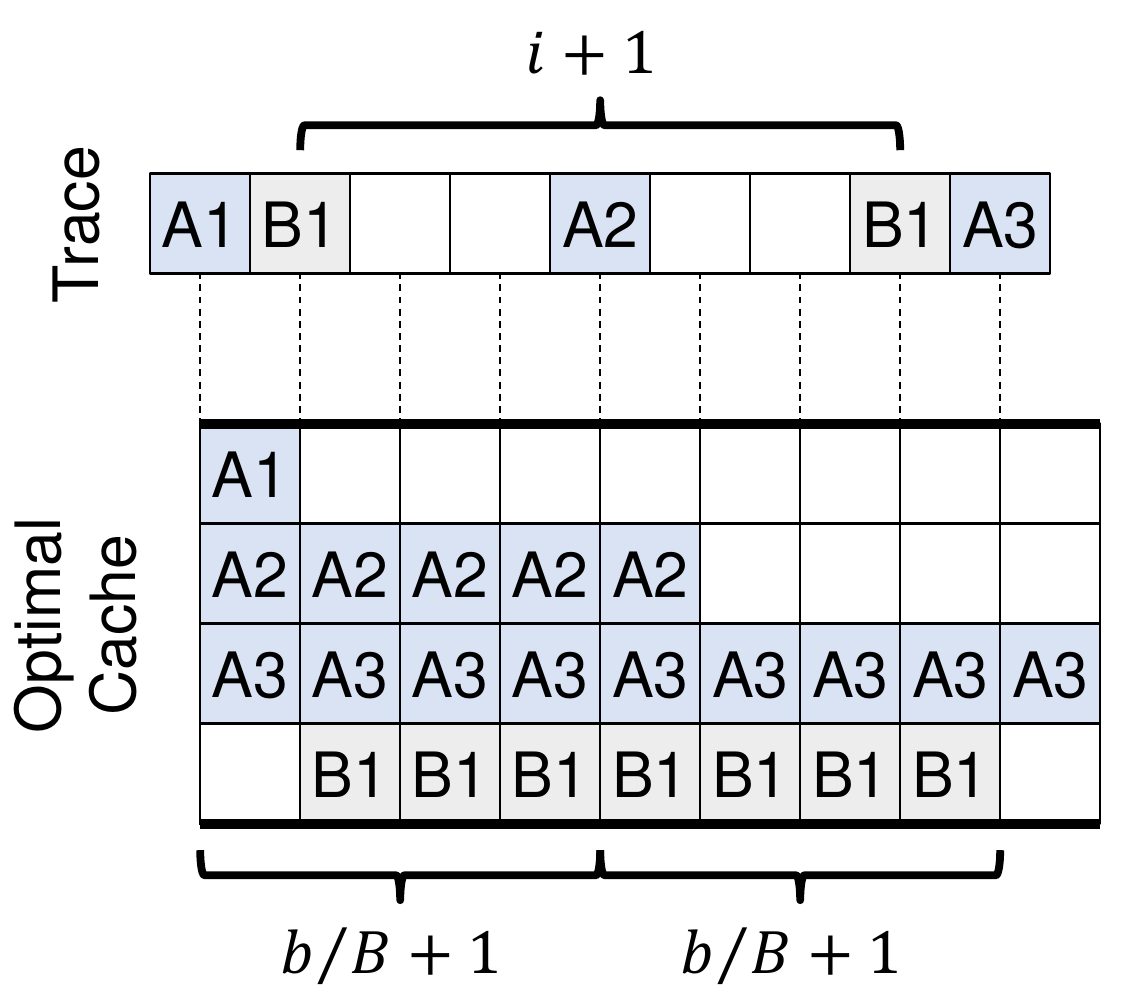}}
\caption{
The pattern of accesses resulting in worst-case traces and the resulting cache 
space used in the optimal cache over time.  
The accesses to A show the worst-case spatial locality, while the accesses to B 
show the worst-case temporal locality. 
}
\label{fig:cache-usage}
\end{figure}

\paragraph{Our Analysis Technique.}
Our analysis visualizes the \opt{}'s performance on a trace as a rectangle, 
with one axis representing the time in units of accesses, 
and the other representing cache space. 
When an \element{} is brought into cache, it take up one unit of space and 
a number of units of time equal to the number of accesses between 
when it is loaded and when it is evicted. 
A visualization of this method can be found in Figure~\ref{fig:cache-usage}. 

In this method, we assume that each access is chosen so as to cause the online 
policy that we compare against to miss.%
\footnote{This assumption results in the worst-case scenario unless accesses to a 
small number of \element{}s can pollute the cache. 
We limit the \blockpart{} to serve only accesses that miss in the \elempart{} 
to prevent such pollution.}
Therefore, if we choose a window of time that accurately captures the average 
long-term behavior of the trace, then the miss ratio is the number of accesses 
in the window $w$ divided by the difference between that number and the number 
of hits $o$ the \opt{} can achieve in the window ($w / (w - o))$. 
We will normalize our result to a window of unit size and consider the hit rate of the \opt{}. 

There are two constraints that limit the number of hits the \opt{} can achieve. 
The first constraint is that of cache usage.
This means that the total number of units of space used by the \opt{} cannot 
exceed the area of the rectangle. 
To reduce clutter, we will ignore the need to leave one unit of cache space 
available for accesses that the \opt{} misses on in this analysis.%
\footnote{The effect of this change on the analysis is limited to replacing 
$h$ with $h-1$ in some places.}
The second constraint is that of accesses. 
This means that the number of accesses that occur cannot exceed the number of 
accesses in the window. 
We will make these constraints more concrete in the analyses below. 

Both of these constraints are necessary to fully specify the problem, but they 
still provide the \opt{} degrees of freedom that are not available in the 
original \pr{} caching problem. 
The constraints limit the total amount of the resource (cache usage or 
accesses) used, but not where these resources are located. 
This allows the \opt{} to make use of solutions that are not viable in the 
original caching problem by having multiple accesses occur at the same time or 
overdrawing cache space at one time in exchange for underutilizing it at a different time. 
Since this looseness empowers the \opt{}, it can only hurt the resulting bounds. 

\paragraph{Temporal Locality.}
We begin our analysis by comparing how the \elempart{} and the \opt{} perform 
on adversarial temporal locality. 
More specifically, we ignore any hits caused by spatial locality. 
In this setting, an access can be a miss for the \elempart{} and a hit for the 
\opt{} if and only if there have been at least \epsize{} distinct \element{}s 
accessed since the \element{} was last accessed. 
This means that any such hit requires at least \epsize{} units of cache 
space. 
Each time this occurs, the access that hits is constrained to be to the same 
\element{} as the access that caused the load. 
The accesses to \block{} B in Figure~\ref{fig:cache-usage} illustrate this pattern. 

Turning these constraints into a linear program and solving it provides an 
upper bound for the competitive ratio of the \elempart{} on temporal locality. 

\begin{theorem}
\label{thm:ub-ep}
When considering hits due to only temporal locality, 
the \elempart{} of \ourpolicy{} has a competitive ratio upper bounded by 
$\frac{\epsizemath{}}{\epsizemath - h}$
where \epsize{} is the size of the \elempart{} cache and 
$h$ is the size of the \opt{}.
\end{theorem}

\begin{proof}
We define $r$ to be the fraction of accesses that the \opt{} hits on due to 
temporal locality. 
The competitive ratio is $\frac{1}{1-r}$. 
Since there is $h$ cache space in the rectangle and each hit requires \epsize{} 
cache space, the cache space constraint is $h \geq{} r \epsizemath{}$.
Since each hit forces one access to a particular \element{}, 
the accesses constraint says that the number of 
hits must be less than the number of accesses, ie $1 \geq r$. 
Although in this problem, this constraint is loose, it becomes critical for 
later versions that include spatial locality. 

This results in the following linear program:
$$\text{Maximize:\;} \frac{1}{1-r}$$
$$\text{subject to:}$$
$$h \geq{} r \epsizemath{}$$
$$1 \geq r$$
$$\text{where } r \text{ is the free variable.}$$

Solving this linear program provides the desired result. 
\end{proof}

This result matches the upper bound from 
Sleator and Tarjan~\cite{sleator1985amortized} for traditional LRU caches.%
\footnote{The lack of a negative one in the denominator is due to the issue 
of space used by misses as discussed earlier.}
Since we are focusing only on temporal locality, this behavior is to be expected, 
as the \elempart{} behaves exactly like an LRU cache of size \epsize{}. 

\paragraph{Spatial Locality.}
We next consider how the \blockpart{} compares to the \opt{} in the face of 
adversarial spatial locality. 
Similar to the analysis above, we will ignore any hits not due to our chosen 
form of locality. 
However, spatial locality introduces several important variations. 

For spatial locality, hits to an \element{} cannot be caused by previous 
accesses to that \element{}, only misses to a different \element{} in their 
\block{} that causes the original \element{} to be loaded. 
This means that each \element{} can cause at most one hit per time that it is 
loaded. 
Furthermore, since the number of \element{}s loaded at a time is upper bounded 
by the block size and one of them is the \element{} that was just missed on, 
the maximum number of hits that can be caused by each miss is at most 
$\bsizemath{} - 1$. 

We now consider what adversarial spatial locality looks like for the 
\blockpart{}. 
In order for the \blockpart{} to miss on an access, there must have been 
$\bpsizemath{} / \bsizemath{}$ distinct other \block{}s accessed since the 
last access to that \block{}. 
This means that in order for the \opt{} to achieve multiple hits from the same 
load operation, each additional \element{} loaded must be stored in cache for 
$\bpsizemath{} / \bsizemath + 1{}$ accesses more than the previous \element{}. 
This results in a triangle-like cache usage pattern, as shown for the 
A \block{} in Figure~\ref{fig:cache-usage}. 

The dimensions of this pattern provide an interesting set of tradeoffs.
Since the \opt{} must miss at least once in order to load \element{}s to hit 
on, the competitive ratio is upper bounded by the number of \element{}s 
loaded at a time. 
However, the cache usage of each \element{} increases with the number of 
\element{}s loaded. 

The design of the linear program is based on the 
number of \element{}s $t$ that \opt{} chooses to load on each miss 
and the number of misses $s$ that cause the \opt{} to perform loads. 
The \opt{} achieves $s (t-1)$ hits, resulting in a competitive ratio of 
$\frac{1}{1-s (t-1)}$. 
Since the total cache usage due to each miss is 
$\sum_{i=0}^{t-1} 1+t(\bpsizemath{}/\bsizemath{} + 1)$, 
the cache usage constraint says that 
$h \geq{} s (\sum_{i=0}^{t-1} 1+t(\bpsizemath{}/\bsizemath{} + 1))$. 
Similarly, the accesses constraint is $1 \geq st$ 
since each load causes the specific miss and each subsequent hit to be 
fixed accesses. 
Solving the resulting linear program provides the resulting bound. 

\begin{theorem}
\label{thm:ub-bp}
When considering hits due to only spatial locality, the 
\blockpart{} of \ourpolicy{} has a competitive ratio upper bounded by 
$\text{min}\Big(\bsizemath{},
(\bpsizemath{} + 2 \bsizemath{} h - \bsizemath{}) / 
(\bpsizemath{} + \bsizemath{}) \Big)$
where \bpsize{} is the size of the \blockpart{} cache and 
$h$ is the size of the \opt{}.
\end{theorem}

\begin{proof}
We define $s$ to be the fraction of accesses where the \opt{} misses and 
performs loads and $t$ to be the number of \element{}s that the \opt{} chooses 
to load on each miss. 
The \opt{} achieves $s (t-1)$ hits, resulting in a competitive ratio of 
$\frac{1}{1-s (t-1)}$. 
Since the total cache usage due to each miss is 
$\sum_{i=0}^{t-1} 1+t(\bpsizemath{}/\bsizemath{} + 1)$, 
the cache usage constraint says that 
$h \geq{} s (\sum_{i=0}^{t-1} 1+t(\bpsizemath{}/\bsizemath{} + 1))$. 
Similarly, the accesses constraint is $1 \geq st$ 
since each load causes the specific miss and each subsequent hit to be 
fixed accesses. 
This results in the following maximization problem:

$$\text{Maximize:\;} \frac{1}{1-s (t-1)}$$
$$\text{subject to:}$$
$$h \geq{} s \Big(\sum_{i=0}^{t-1} 1+t(\bpsizemath{}/\bsizemath{} + 1)\Big)$$
$$1 \geq st$$
$$\text{where } s \text{ and } t \text{ are the free variables.}$$

Solving this maximization problem provides the second term in the theorem. 
The first term comes from the fact that $t$ cannot exceed \bsize{} combined 
with the second constraint in the linear program. 
\end{proof}

\paragraph{Combining Temporal and Spatial Locality.}
We now show how to combine the two methods of achieving hits to obtain 
an upper bound for the entirety of \ourpolicy{} for general traces. 
Since \ourpolicy{} hits on an access if either of its partitions hits on that 
access, the restrictions on an access in order for it to be a miss for 
\ourpolicy{} must be stricter. 
This allows us to formulate a linear program for the entire policy by combining 
the hits and the constraints of the previous two versions. 

The resulting linear program is too complex to solve directly. 
To deal with this, we modify the spatial locality problem to take the number 
of hits due to temporal locality as an input. 
We then use the result of this problem to choose the number of temporal 
locality hits that maximizes the competitive ratio. 
The result of this is shown below in Theorem~\ref{thm:iblp-ub}.

\begin{theorem}\label{thm:iblp-ub}
The competitive ratio of \ourpolicy{} is upper bounded by:
$$\begin{cases}
  \frac{(\bpsizemath{} + \bsizemath{} (2 i - 1))^2}{8 \bsizemath{} (
      \bsizemath{} + \bpsizemath{}) (\epsizemath{} - h)} & 
      \epsizemath{} \leq{} \frac{2\bsizemath{}\bpsizemath - \bpsizemath{} + 
      2\bsizemath{}^2 + \bsizemath{}}{2\bsizemath{}} \\
  \frac{2\bsizemath{}\epsizemath{} - \bsizemath{}\bpsizemath{} + 
      \bpsizemath{} - \bsizemath{}^2 - \bsizemath{}}{2\epsizemath{} - 2h} &
      \epsizemath{} > \frac{2\bsizemath{}\bpsizemath - \bpsizemath{} + 
      2\bsizemath{}^2 + \bsizemath{}}{2\bsizemath{}} 
\end{cases}$$
where $\epsizemath{} \geq h$ is the size of the \elempart{}, 
\bpsize{} is the size of the \blockpart{}, 
and $h$ is the size of the \opt{}. 
\end{theorem}

\begin{proof}
As in the previous proofs, 
we define $r$, $s$, and $t$ to be the fraction of accesses that the \opt{} hits 
on due to temporal locality, the fraction of accesses that the \opt{} misses on 
and loads \element{}s for spatial locality, and the number of \element{}s 
loaded for spatial locality, respectively. 

The total number of hits that \opt{} achieves is equal to the sum of the 
hits from the individual localities. 
This results in a competitive ratio of $\frac{1}{1-r-s(t-1)}$. 
We combine the amount of cache space used and the number of accesses 
forced in a similar way. 
This results in the following linear program:

$$\text{Maximize:\;} \frac{1}{1-r-s(t-1)}$$
$$\text{subject to:}$$
$$h \geq{} r \epsizemath{} + s \Big(\sum_{i=0}^{t-1} 1+t(\bpsizemath{}/\bsizemath{} + 1)\Big)$$
$$1 \geq r + st$$
$$\text{where } r \text{, } s \text{, and } t \text{ are the free variables.}$$

Unfortunately, we were unable to solve this linear program directly 
(one hour of computation time in Wolfram Mathematica). 
To obtain a solution, we break the program into smaller chunks that can 
be solved individually. 
We start by modifying the linear program to compute the values of $s$ and $t$ 
that maximize the competitive ratio given a particular $r$ value. 
This results in the following values for $s$ and $t$:

$$s =  \frac{(\bsizemath{} + \bpsizemath{})(1-r)^2}{\bpsizemath{} - 
\bpsizemath{} r + \bsizemath{} (2 h -1+ r - 2 \epsizemath{} r)}$$
$$t = \frac{\bpsizemath{} - \bpsizemath{}r + \bsizemath{}(2 h -1+ r - 2 
\epsizemath{} r)}{(\bsizemath{} + \bpsizemath{}) (1 - r)}$$

When we plug in these values and solve the resulting maximization problem, 
we achieve results for both $r$ and the competitive ratio: 


$$r = \frac{\bpsizemath{} + \bsizemath{} (4 h - 2 \epsizemath{} - 1)}{
    \bpsizemath{} + \bsizemath{} (2 \epsizemath{} - 1)}$$
$$\text{Ratio} = \frac{(\bpsizemath{} + \bsizemath{} (2 i - 1))^2}{8 \bsizemath{} (
    \bsizemath{} + \bpsizemath{}) (\epsizemath{} - h)}$$

This result is a valid upper bound, but fails to account for the constraint 
that $t$ cannot exceed the \block{} size \bsize{}. 
By using the expressions for the values of $r$ and $t$, we find that this 
occurs when 
$\epsizemath{} > \frac{2\bsizemath{}\bpsizemath - \bpsizemath{} + 
    2\bsizemath{}^2 + \bsizemath{}}{2\bsizemath{}}$.
In this region, we know from above that $t$ maxes out at \bsize{}. 
We apply this change to the prior analysis, with the following results:

$$r = \frac{2\bsizemath{}h - \bsizemath{}\bpsizemath{} + \bpsizemath{} - \bsizemath{}^2 - \bsizemath{}}{
    2\bsizemath{}\epsizemath{} - \bsizemath{}\bpsizemath{} + \bpsizemath{} - \bsizemath{}^2 - \bsizemath{}}$$
$$\text{Ratio} = \frac{2\bsizemath{}\epsizemath{} - \bsizemath{}\bpsizemath{} + 
    \bpsizemath{} - \bsizemath{}^2 - \bsizemath{}}{2\epsizemath{} - 2h}$$

Putting these results together finishes the proof. 
\end{proof}

\subsection{Applying the Bound}\label{sec:ub-app}

Having proved a bound on the competitive ratio for \ourpolicy{} as a function 
of the layer sizes, we can can now consider how to partition the cache space. 
This task is complicated by the fact that the optimal partitioning depends on 
the size of the optimal cache being compared against. 

\paragraph{Known optimal size.}
When the size of the optimal cache is known, the optimal layer sizes 
can be directly computed. 
When 
$k \geq (3 \bsizemath{} h - h - \bsizemath{}^2 - \bsizemath{})/(\bsizemath{} - 1)$, 
this results in: 
$$\text{Ratio} = \frac{(k + \bsizemath{} - 1) (k - h + \bsizemath{} (2 h - 1) )}{(k - h + \bsizemath{})^2}$$ 
$$\epsizemath = \frac{k^2 + 4 \bsizemath{} h k - h k + 4 \bsizemath{}^2 h - 
    3 \bsizemath{} h - \bsizemath{}^2}{2 \bsizemath{} k + k + 2 \bsizemath{} h - 
    h + 2 \bsizemath{}^2 - 3 \bsizemath{}}$$ 
For smaller $k$ values, setting $i = k$ (i.e., operating as an \elempol{}) provides the minimum competitive ratio 
of: 
$$\frac{2 \bsizemath{} k - \bsizemath{}^2 - \bsizemath{}}{2(k-h)}$$
This transition occurs at the point where the competitive ratio due to temporal 
locality exceeds the maximum competitive ratio that can be achieved due to 
spatial locality ($\bsizemath{} \times$). 
In other words, for small $k$ values (relative to $h$), temporal locality dominates performance. 

Again considering large caches with large \block{}s
($k > h \gg \bsizemath{} \gg 1)$,
we see that the ratio is roughly 
$\frac{k(k+2\bsizemath{}h)}{(k-h)^2}$ if $k \geq{} 3h$ 
and $\frac{\bsizemath{}k}{k-h}$ if $k < 3h$. 

Figure~\ref{fig:bound-comp-graph} shows how this upper bound 
compares to the lower bound, as well as single-granularity caches 
of the same size running LRU. 
\ourpolicy{} outperforms the small-granularity Item Cache for $k \approx{} 3h$ and 
larger, and it outperforms the large-granularity Block Cache for 
$k \approx{} 4\bsizemath{}h$ and smaller.
In addition, \ourpolicy{} performs close to optimal for all values of $k$, 
whereas the performance of the baselines degrades severely outside of their 
ideal performance conditions. 

Table~\ref{tab:int-points} shows how this bound compares to traditional 
caching and our lower bound. 
Our upper bound has roughly the same penalty to 
augmentation and competitive ratio as our lower bound, 
differing by at most a multiplicative factor of $3\times$. 
Comparing against the points of interest from the lower bound, we see that 
the competitive ratio is $\approx 2\bsizemath{}$ when $k = 2h$, 
$k \approx \bsizemath{}h$ yields a competitive ratio of $\approx 3$, 
and the meeting point occurs when $k \approx \sqrt{2\bsizemath{}}h$. 

\begin{figure}
\includegraphics[width=\linewidth]{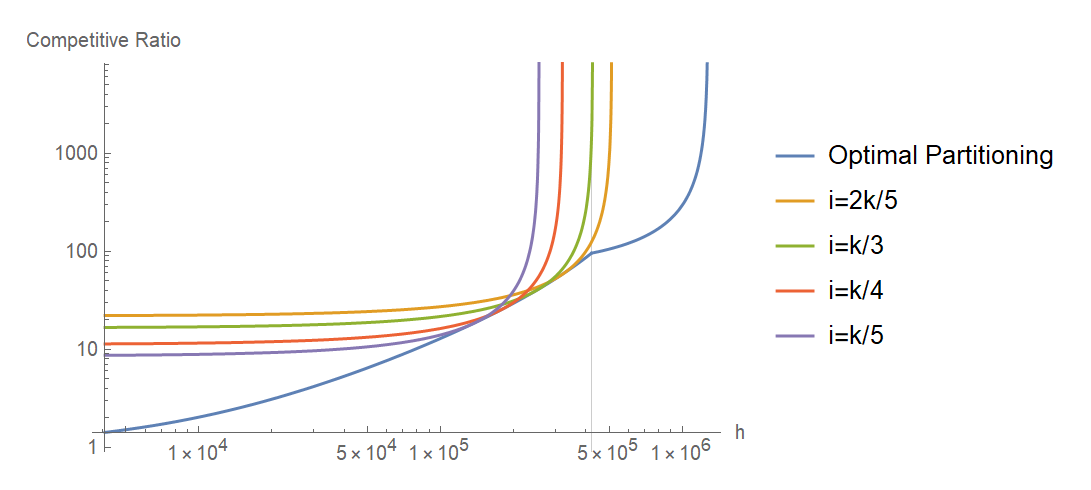}
\caption{
A graph showing how the upper bound of \ourpolicy{} with 
constant layer sizes performs compared to the optimal layer sizes. 
The x-axis is the size of the optimal cache and the y-axis is the competitive 
ratio (lower is better). 
In this graph, the online cache size $k = 1.28M$ and \block{} size 
$\bsizemath{} = 64$. 
}
\label{fig:fixed-size-comp}
\end{figure}

\paragraph{Unknown optimal size.}
The optimal layer sizes in IBLP depend on the size of the optimal cache $h$
being compared against. 
For any fixed layer sizes, the competitive ratio will be optimal
at only one value of $h$, but show significant degradation for 
larger $h$ and limited improvement for smaller $h$ 
(see Figure~\ref{fig:fixed-size-comp}). 

This dependency on the size of the optimal cache 
is unique amongst known caching problems. 
Unlike other caching problems, 
the competitive ratios due to temporal locality and spatial locality 
are different functions of the optimal cache size. 
As a result, the relative performance of traces changes depending on optimal 
cache size, which results in different values of $h$ having 
different worst-case traces. 
This suggests that a full understanding of the \problem{} 
requires analysis beyond competitive ratios.

\section{Analysis of Randomized Algorithms}
\label{sec:rand}

Our previous analysis showed that the relative competitiveness of different replacement policies may depend on the size of the policy compared against. 
This is undesirable, since we would like to be able to compare policies without relying on a hypothetical comparison point. 
In this section, we explore the randomized replacement policies in the \problem{}.
One would hope that randomized policies, which have can more easily compare against
Unfortunately, this hope proves false; we are not able to eliminate the impact of comparison size on relative competitiveness. 

\subsection{A Randomized Policy}
\label{sec:rand-pol-desc}

We design a randomized policy by extending the ideas of marking algorithms to the \problem{}. 
As with traditional marking algorithms, \element{}s can be either marked or unmarked. 
We mark \element{}s when they are requested. 
Evictions must choose from unmarked \element{}s, unless all \element{}s are marked, in which case all marks are removed and then an \element{} is evicted at random. 

Our policy, known as \ourrandpolicyfull{} (\ourrandpolicy{} for short), accounts for granularity change by loading each \element{} in the accessed \block{}, but \emph{not} marking them. 
This results in a system where \element{}s that show spatial but not temporal locality are loaded into cache but do not result in the eviction of \element{}s with temporal locality. 
In the special (but common) case where the number of unmarked \element{}s in the cache is smaller than the \block{} size but greater than zero, the requested \element{} is loaded and then the remaining unmarked \element{}s in cache are replaced by randomly selected \element{}s from the accessed \block{}. 

We will briefly compare \ourrandpolicy{} with similar policies that make different choices for handling granularity change. 
First, we compare to a marking algorithm that ignores granularity change. 
It is easy to show that this policy has a competitive ratio of at least \bsize{} regardless of size (as long as the \opt is \bsize{} or larger) by repeatedly choosing a new \block{} and accessing each \element{} in it. 
This shows the importance of loading additional \element{}s from the requested \block{}. 
However, a policy that loads and marks every \element{} in the \block{} also has issues. 
In particular, like the analysis of \blockpol{}s in Section~\ref{sec:block-pol-lb}, when the trace does not provide spatial locality, the effective size of the cache is reduced by the excess \element{}s. 
Unlike the deterministic setting, there may be value in a policy that loads some but not all of the \element{}s in the accessed \block{}. 
However, due to the issues with competitive ratios in this setting, we did not pursue this analysis in detail. 

\subsection{Analyzing Relative Competitiveness}
\label{sec:rand-rel}

We now show that randomized algorithms in the \problem{} do not bypass the issue where their relative competitiveness may depend on the size of the cache compared against. 
In particular we continue to have the issue where policies must decide how many \element{}s to load on each access. 
Similar to the above analysis, caches that load few \element{}s perform well compared to \opt{}s of similar size, when each line devoted to temporal locality matters significantly. 
However, this hinders their performance when compared against \opt{}s of significantly smaller size, where spatial locality becomes more important. 
Caches that load many \element{}s on each access suffer from the opposite set of tradeoffs. 
This relative change is not affected by randomness, leaving our issue unresolved. 

\section{Beyond Competitive Ratios}
\label{sec:beyond-comp-rat}

We would like to be able to use the \problem{} to design and analyze caches independent of a hypothetical comparator. 
However, competitive ratios seem to require this knowledge to provide insight into policy performance. 
To remedy this, we turn to the locality of reference model discussed in Section~\ref{sec:model}. 
In this model, the functions $f(n)$ and $g(n)$ provide bounds on the number of \element{}s and \block{}s, respectively, seen in a window of $n$ accesses. 
We use these functions to provide bounds on the fault rates of deterministic policies in the \problem{}. 


\subsection{Lower Bound}
The addition of spatial locality in the \prob{} means that the lower bound on fault rates in the traditional locality of reference model does not apply. 
We provide a lower bound for deterministic policies in this new model. 


\begin{theorem}\label{thm:lor-lb}
Any deterministic replacement policy has a fault rate of at least: 
$$ \frac{g(f^{-1}(k+1) - 2)}{f^{-1}(k+1) - 2} $$
where $k$ is the size of the cache. 
\end{theorem}


\begin{proof}
Following the ideas of Albers et al.~\cite{albers2005paging}, we will construct a family of traces such that for any $n$, there exists a trace of length $n$ in the family where the policy will have a fault rate no lower than the bound. 
We construct traces in our family using $k+1$ distinct \element{}s. 
Due to the locality constraints, these \element{}s can be partitioned into at most $g(f^{-1}(k+1) - 2)$ \block{}s. 
We generate traces in phases, where each phase consists of $f^{-1}(k+1) - 2$ accesses divided into $k-1$ repetitions. 
A repetition consists of repeated accesses to a single \element{} that has not yet been accessed this phase. 
In each phase, repetition $1 \leq{} j \leq{} k-1$ starts with the $f^{-1}(j+1) - 1$th access of that phase and continues until the access before the next block starts. 
We can apply the techniques of Albers et al.~\cite{albers2005paging} to show that these traces are consistent with $f(n)$. 


It remains to discuss $g(n)$ and show the minimum fault rate of policies on these traces. 
In the work of Albers et al.~\cite{albers2005paging}, the \element{} chosen is the one of the $k+1$ used in the trace that is not in the cache, and therefore every repetition causes one miss. 
However, in our model, our ability to choose \element{}s is limited by the $g(n)$ function. 
In particular, the \element{} not in cache may be from a \block{} that has not been accessed yet in the phase. 
If this is the case, then choosing that \element{} increases the number of \block{}s accessed in the window, which may cause a violation. 
However, it is known that a new \block{} can be chosen at least $g(p)$ times for a phase of length $p$. 
Each of these new \block{} choices returns us the freedom to guarantee that we can pick the \element{} that is not in cache. 
\end{proof}

\subsection{Upper Bound}
We now provide an upper bound for the fault rate of \ourpolicy{}. 
Since this policy consists of two layers of cache acting in concert, both layers must miss in order for the entire cache to miss.
We therefore begin by providing bounds for both layers individually. 

\begin{theorem}
The fault rate of the \elempart{} of \ourpolicy{} is at most:
$$ \frac{\epsizemath{} - 1}{f^{-1}(\epsizemath{} +1) - 2} $$
where \epsize{} is the size of the \elempart{}. 
\end{theorem}

\begin{proof}
The \elempart{} is simply an LRU cache that operates as though it were in the traditional model. 
The change in model cannot cause the fault rate to increase, since it only introduces new ways for a policy to hit. 
This means that we can rely on the result from Albers et al.~\cite{albers2005paging} on bounds on LRU policies in the traditional model.
Plugging in the size of the \elempart{} as the cache size to this provides our result. 
\end{proof}

The upper bound for the fault rate of the \blockpart{} requires a slight transformation before relying on the same straightforward application of the traditional bound. 

\begin{theorem}
The fault rate of the \blockpart{} of \ourpolicy{} is at most:
$$ \frac{\bpsizemath{} / \bsizemath{} - 1}{f^{-1}(\bpsizemath{} / \bsizemath{} +1) - 2} $$
where \epsize{} is the size of the \elempart{}. 
\end{theorem}

\begin{proof}
The \blockpart{} is a \blockpol{}, meaning that it loads and evicts at the granularity of \block{}s. 
Thus, the behavior of the \blockpart{} depends on the \block{} that is accessed, but is independent of the particular \element{}. 
We can combine these facts to view the \blockpart{} as an LRU cache with effective size $\bpsizemath{} / \bsizemath{}$ serving a trace where requests are to \block{}s. 
Substituting in the effective size of the cache and using the number of \block{}s in a window $g(n)$ as the \element{}s per window function, we achieve the resulting bound. 
\end{proof}

Taking the minimum of these fault rates for a given input provides an upper bound on the fault rate of \ourpolicy{}. 

\begin{theorem}
The fault rate of \ourpolicy{} with \elempart{} size \epsize{} and \blockpart{} size \bpsize{} is upper bounded by: 
$$ \min{} \left( \frac{\epsizemath{} - 1}{f^{-1}(\epsizemath{} +1) - 2}, \frac{\bpsizemath{} / \bsizemath{} - 1}{f^{-1}(\bpsizemath{} / \bsizemath{} +1) - 2} \right)$$
where $f(n)$ is a function mapping a window size to the maximum number of \element{}s accessed in a window of size $n$
and $f(n)$ is a function mapping a window size to the maximum number of \element{}s accessed in a window of size $n$. 
\end{theorem}

\begin{table}[t]
  \footnotesize
  \resizebox{\linewidth}{!}{%
    \newcommand{\headerfont}[1]{#1}
\newcommand{\entryfont}[1]{{\normalsize #1}}
\begin{tabular} {@{}cc|ccc@{}}
\toprule
\bf f(n) & \bf g(n) & \bf Lower Bound & \bf \elempart{} UB & \bf \blockpart{} UB \\
\midrule
\headerfont{$x^{1/2}$} & \headerfont{$x^{1/2}$} & \entryfont{$1/h$}   & \entryfont{$1/\epsizemath{}$} & \entryfont{$\bsizemath{}/\bpsizemath{}$} \\
\headerfont{$x^{1/2}$} & \headerfont{$x^{1/2}/\bsizemath^{1/2}$} & \entryfont{$1/(\bsizemath^{1/2}h)$}   & \entryfont{$1/\epsizemath{}$} & \entryfont{$1/\bpsizemath{}$} \\
\headerfont{$x^{1/2}$} & \headerfont{$x^{1/2}/\bsizemath{}$} & \entryfont{$1/\bsizemath{}h$}   & \entryfont{$1/\epsizemath{}$} & \entryfont{$1/\bsizemath{}\bpsizemath{}$} \\ 
\headerfont{$x^{1/p}$} & \headerfont{$x^{1/p}$} & \entryfont{$1/h^{p-1}$}   & \entryfont{$1/\epsizemath^{p-1}$} & \entryfont{$\bsizemath^{p-1}/\bpsizemath^{p-1}$} \\
\headerfont{$x^{1/p}$} & \headerfont{$x^{1/p}/\bsizemath^{1/2}$} & \entryfont{$1/(\bsizemath^{(p-1)/p}h^{p-1})$}   & \entryfont{$1/\epsizemath^{p-1}$} & \entryfont{$1/\bpsizemath^{p-1}$} \\
\headerfont{$x^{1/p}$} & \headerfont{$x^{1/p}/\bsizemath{}$} & \entryfont{$1/\bsizemath{}h^{p-1}$}   & \entryfont{$1/\epsizemath^{p-1}$} & \entryfont{$1/\bsizemath{}\bpsizemath^{p-1}$} \\ 
\bottomrule
\end{tabular}
    }
  \caption[Salient bounds]{Salient bounds
    for comparing an equally split cache ($\epsizemath{} = \bpsizemath{}$) to the lower bound for a cache of half the size ($h = \epsizemath + \bpsizemath{}$).}
\label{tab:lor-points}
\end{table} 

\subsection{Analysis}
We have provided a method for analyzing the performance \ourpolicy{} or other policies in the \prob{} that does not rely on the existence of a hypothetical optimal comparison point. 
We now show how to apply this analysis. 
To do this, we will analyze \ourpolicy{} with equal partition sizes ($\epsizemath{} = \bpsizemath{}$).

For this analysis, we will consider polynomial locality functions (i.e., $f^{-1}(n) = cn^{p}$ for some real numbers $c$ and $p$). 
Since locality functions must be positive concave functions, this covers the majority of high order terms that would occur in real traces. 
Exponential functions are also possible, but such functions signify an extremely high degree of locality, where even small, simple caches should perform well. 
We will compare our result with the lower bound for a cache of the same size as each partition, which provides a result roughly equivalent to competitive ratio with an augmentation factor of $2\times$ (i.e., $\epsizemath{} + \bpsizemath{} = k = 2h$).


When $f(n) = g(n)$, there is no spatial locality. 
In this situation, the upper bound for the \elempol{} matches the lower bound of the baseline, while the \blockpol{} will be off by a factor of $\bsizemath{}^{p-1}$ (recall that $p$ is the degree of the high order term of $f(n)$). 
With maximum spatial locality, $f(n) = \bsizemath{} g(n)$. 
In this situation, the upper bound for the \blockpol{} matches the lower bound of the baseline, while the \elempol{} will be off by a factor of $\bsizemath{}$. 
The largest gap between the baseline and the upper bound for \ourpolicy{} occurs when the ratio between $f(n)$ and $g(n)$ is $\bsizemath^{1-(1/p)}$. 
This represents input with very high, but not maximal, spatial locality. 
For input with this locality, the upper bounds for both partitions meet at $1 / \epsizemath^{p-1}$. 
This ends up differing from the baseline by a factor equal to the ratio of $f(n)$ and $g(n)$, $\bsizemath^{1-(1/p)}$. 
As the value of $p$ approaches $\infty$, this gap approaches \bsize{}. 

There are two primary takeaways from this analysis. 
The first is that the two methods of analysis are in agreement, as our competitive ratio results also show a ratio of \bsize{} for an augmentation factor of 2. 
The second is that the performance of \ourpolicy{} is worst when locality is high (large $p$ implies high temporal locality and large ratios of $f(n)$ and $g(n)$ implies high spatial locality). 
When locality is so high, our miss rate will be very low. 
Therefore, the relatively large multiplicative factor gap will result in a small number of additive misses. 
This suggests that \ourpolicy{} will perform well in practice.



\section{Conclusion}
\label{sec:conc}

In this work, we have provided a theoretical foundation for the study of spatial locality in caching. 
We started with a traditional treatment of the new \problem{}, 
showcasing the changes caused by spatial locality through
proving an NP-completeness result for the offline problem 
and new lower bounds on competitive ratios. 
We used the insights gained from these works to develop \ourpolicy{} 
and provide strong, provable bounds for its performance. 
Our upper bound illustrates how competitive ratios become a flawed metric in the \problem{} due to a dependency on the size of the hypothetical comparison point. 
We solve this problem by extending a locality of reference model to the \problem{}, 
and provide analysis of \ourpolicy{} in this new model. 

\paragraph{Acknowledgments.}
  Supported in part by NSF grants
  CCF-1919223, CCF-2028949, 
  a Google Research Scholar Award, 
  a VMware University Research Fund Award, and by the Parallel Data Lab (PDL)
  Consortium. 

%
\bibliographystyle{abbrv}
\bibliography{ref}

%

\end{document}